\documentclass[noshowpacs,preprintnumbers,amsmath,amssymb,superscriptaddress,onecolumn]{revtex4}
\usepackage[utf8]{inputenc}
\usepackage[english]{babel}
\usepackage{amsmath}
\usepackage{amsfonts}
\usepackage{amssymb}
\usepackage{xcolor}
\usepackage{mathrsfs}
\usepackage{amsthm}
\usepackage{graphicx,changepage}

\newcommand{\be}{\begin{equation}}
\newcommand{\ee}{\end{equation}}

\theoremstyle{plain}
\newtheorem{Remark}{Remark}
\newtheorem{Theorem}{Theorem}
\newtheorem{Lemma}{Lemma}
\newtheorem{Proposition}{Proposition}

\newtheorem{Definition}{Definition}

\begin{document}

\title{The {\em Relativistic} Hopfield network: rigorous results}

\author{Elena Agliari}
\affiliation{Dipartimento di Matematica, Sapienza Universit\`a di Roma, Italy}
\author{Adriano Barra}
\affiliation{Dipartimento di Matematica e Fisica Ennio De Giorgi, Universit\`a del Salento, Italy}
\author{Matteo Notarnicola}
\affiliation{Dipartimento di Matematica, Sapienza Universit\`a di Roma, Italy}
\affiliation{Dipartimento di Matematica e Fisica Ennio De Giorgi, Universit\`a del Salento, Italy}

\begin{abstract}The relativistic Hopfield model constitutes a generalization of the standard Hopfield model that is derived by the formal analogy between the statistical-mechanic framework embedding neural networks and the Lagrangian mechanics describing a fictitious single-particle motion in the space of the tuneable parameters of the network itself. In this analogy the cost-function of the Hopfield model plays as the standard kinetic-energy term and its related Mattis overlap (naturally bounded by one) plays as the velocity. The Hamiltonian of the relativisitc model, once Taylor-expanded, results in a $P$-spin series with alternate signs: the attractive contributions enhance the information-storage capabilities of the network, while the repulsive contributions allow for an easier unlearning of spurious states, conferring overall more robustness to the system as a whole.
\newline
Here we do not deepen the information processing skills of this generalized Hopfield network, rather we focus on its statistical mechanical foundation. In particular, relying on Guerra's interpolation techniques, we prove the existence of the infinite volume limit for the model free-energy and we give its explicit expression in terms of the Mattis overlaps. By extremizing the free energy over the latter we get the generalized self-consistent equations for these overlaps, as well as a picture of criticality that is further corroborated by a fluctuation analysis.
\newline
These findings are in full agreement with the available previous results.
\end{abstract}


\maketitle

\section{Introduction: a glance at the {\em relativistic} Hopfield network}

Recent advances, in hardware (mainly due to the novel generation of GPU computing architectures over the standard CPU ones \cite{GPU}) as well as in software (mainly due to the novel generation of algorithmic prescriptions overall termed {\em Deep Learning} \cite{DLbook}), have made neural networks pervasive in every-day life and this obviously raised the quest for stronger mathematical frameworks where these algorithms and models can be suitably analyzed, controlled or even {\em understood} \cite{DL1}.
\newline
To this goal, statistical mechanics has proved to be a particularly convenient tool, as evidenced by the seminal work by John Hopfield \cite{Hopfield} and the successive analysis carried out by Amit-Gutfreund-Somponlinsky \cite{Amit,Coolen,Viktor}, and it will constitute the main tool exploited in this paper too.
\newline
In the remaining of this introductory section, we briefly revise the classical and the relativistic Hopfield models, pointing out their capabilities as neural networks, within a statistical mechanical setting. Then, in the next two sections, we focus on the relativistic generalization and, in particular, in Sec.~\ref{sec:2} we prove the existence of the thermodynamic limit for its free energy\footnote{Notice that, once the existence of the thermodynamic limit is proved for the free energy, it holds, in a cascade fashion, also for other various quantities of interest, as entropy and internal energy \cite{Guerra-LimTerm1}.}, while in Sec.~\ref{sec:3} we give an explicit expression for such a quantity (that turns out to coincide sharply with the one obtained with mechanical techniques in \cite{Albert}) in terms of the natural order parameters of the theory, namely the Mattis overlaps; moreover, by extremizing the free energy over the Mattis overlaps we obtain self-consistent expressions for their evolution that is further confirmed by a study of the fluctuations of the (rescaled and centered) Mattis overlaps (to inspect ergodicity breaking and the critical behavior of the system). Finally, Sec.~\ref{sec:4} is left for our conclusions and outlooks.

\subsection{The {\em classical} Hopfield network in a nutshell}

Consider $N$ Boolean (i.e., Ising) spins (or {\em neurons}) $\sigma_i = \pm 1$, $i \in (1,...,N)$ and $P$ patterns $\xi^{\mu}$, $\ \mu \in (1,...,P)$ of length $N$ whose entries are identically and independently drawn from
\be \label{eq:iid}
P(\xi_i^{\mu}=+1)=P(\xi_i^{\mu}=-1)=1/2.
\ee
Note that  we are taking the pattern entries completely at random: this choice may sound poorly realistic, yet it ensures, by a Shannon compression argument, that if the network is able to deal with $P$ entirely random patterns, it will be likewise able to deal with -at least- $P$ structured ones.
\begin{Definition}
The cost function (or {\em Hamiltonian} to keep a physical jargon) of the classic Hopfield model can be written as
\be\label{Hopfield-classico}
H_N^{\textrm{(c)}}(\boldsymbol \sigma| \boldsymbol \xi) :=-\frac{1}{2N}\sum_{i,j}^{N,N}\left(\sum_{\mu=1}^{P}\xi_i^{\mu}\xi_j^{\mu}\right) \sigma_i \sigma_j.
\ee
\end{Definition}
In order to have an intuition about the {\em spontaneous associative memory} capabilities of this model, one quantifies the retrieval pertaining to each pattern and this is accomplished by the following order parameters
\begin{Definition}
For each stored pattern of information $\xi^{\mu}$, we introduce the Mattis overlap $m_{\mu}$ defined as
\be
m_{\mu} := \frac{1}{N}\sum_{i=1}^{N}\xi_i^{\mu}\sigma_i \in [-1+1].
\ee
\end{Definition}
\noindent
Notice that the Hopfield Hamitonian can then be written as $H_N^{\textrm{(c)}}(\boldsymbol \sigma| \boldsymbol \xi)=- N \sum_{\mu=1}^{P} m_{\mu}^2$. This expression highlights that, if the neural configuration $\boldsymbol \sigma $ is uncorrelated with any of the $P$ patterns, the scalar product of the vector state $(\sigma_1, \sigma_2,...,\sigma_N)$ with any of these patterns, say $\boldsymbol{\xi}^{\mu} = (\xi_1^{\mu}, \xi_2^{\mu},...,\xi_N^{\mu})$, would vanish as $N^{-1/2}$, hence its contribution to lower the cost function (namely the {\em energy of the system}) would be rather marginal; one the other hand, if the neural configuration is highly correlated with one of the patterns, then its corresponding Mattis overlap would be $O(1)$ and this would significantly decrease the energy of the system: as the Hamiltonian is parabolic in the Mattis overlap, this observation candidates the $P$ patterns to act as {\em attractors} for any reasonable network's dynamics \cite{Coolen}. Therefore, if the network is fed with partial information concerning one pattern (for instance a corrupted pattern is presented to the network\footnote{For instance, information can be supplied to the network as an {\em external field} to keep a physical jargon.}), it will autonomously denoise the supplied information and find out the correct (pure) reference pattern (if the noise level is not too high \cite{Amit}).
It is worth noticing that, according to (\ref{eq:iid}), patterns become orthogonal in the infinite volume limit in such a way that the system can retrieve patterns only sequentially (see \cite{Agliari-PRL1,Agliari-PRL2,Monasson-PRL} for examples of neural networks able to perform in a parallel way).
\newline
An extensive mathematical formalization of these statements requires the introduction of several concepts and goes beyond the scope of the present paper: we refer to excellent textbooks for deepening the associative memory capabilities of the Hopfield network (see e.g. \cite{Amit,Coolen,Viktor}) and to further readings for understanding the mathematical complexity behind these models \cite{Guerra1,Guerra2,Bovier1,Bovier2,Tala1,Tala2,Agliari-Barattolo}.
\newline
An important remark is that, once introduced the {\em storage} of the network as the limiting ratio $\alpha= \lim_{N \to \infty} P/N$, we distinguish between two different regimes: the low-storage case, where $\alpha=0$, and the high-storage case, where $\alpha>0$. As for the latter, at a rigorous level, several questions still need to be answered (e.g., a proof of the existence of the infinite volume limit for its free energy at present is still lacking, not to mention a full, clear replica-symmetry-breaking picture for this type of {\em spin glass}), thus generalizations of the high-storage case may be still premature. In fact, the relativistic generalization of the Hopfield model was introduced as a low-storage model in \cite{Albert}, and here we provide a rigorous treatment of the generalization still at $\alpha = 0$.

\subsection{The {\em relativistic} Hopfield network in a nutshell}
As discussed and investigated in \cite{andrei,danis}, the analysis of the Hopfield model (and suitable extensions) can be mapped into a mechanical problem. More precisely, the free-energy associated to the model can be shown to fulfill an Hamilton-Jacobi equation describing a fictitious single-particle motion in a $P$-dimensional space, 
where the Mattis overlap plays as the particle velocity, which therefore displays an intrinsic bound ($m_{\mu} \leq 1, \forall \mu$). When studying the associative memory capabilities of the Hopfield model, one is particularly interested in regimes where the velocity is approaching its upper bound, in such a way that the correct framework to embed the problem is the relativist one and the related generalized Hopfield model is described hereafter.

Consider $N$ Boolean (i.e. Ising) spins (or {\em neurons}) $\sigma_i = \pm 1$, $i \in (1,...,N)$ and $P$ patterns $\xi^{\mu}$, $\ \mu \in (1,...,P)$ of length $N$ whose entries are independently and identically drawn from
$$
P(\xi_i^{\mu}=+1)=P(\xi_i^{\mu}=-1)=1/2.
$$
\begin{Definition}
The cost function (or {\em Hamiltonian} to keep a physical jargon) of the relativistic Hopfield model can be written as
\be\label{Hopfield-relativistico}
H_N^{\textrm{(r)}}(\boldsymbol \sigma| \boldsymbol \xi) := - N \sqrt{1+ \sum_{\mu=1}^{P}m_{\mu}^2} \equiv  - N \sqrt{1+ \bold{m}^2}.
\ee
\end{Definition}
Note that this cost function (\ref{Hopfield-relativistico}) can be expanded in an alternate-sign series as
\be\label{series}
\frac{H_N^{\textrm{(r)}}(\boldsymbol \sigma| \boldsymbol \xi)}{N} \sim -1 - \frac{1}{2N}\sum_{i,j}^{N,N}(\boldsymbol{\xi}_i\cdot \boldsymbol{\xi}_j)\sigma_i \sigma_j + \frac{1}{8 N^4} \sum_{i,j,k,l}^{N,N,N,N}(\boldsymbol{\xi}_i\cdot \boldsymbol{\xi}_j)(\boldsymbol{\xi}_k\cdot \boldsymbol{\xi}_l)\sigma_i \sigma_j \sigma_k \sigma_l + O(\sigma^6),
\ee
thus, focusing on the attractive contributions (beyond the classical pairwise model $P=2$), it is enriched by $P$-spin terms (with $P=6, 10, ...$) that yield to further synaptic couplings where information can be stored (as recently suggested by Hopfield himself \cite{Dimitry}), while, focusing on the repulsive contributions, it also displays $P$-spin terms (with $P=4, 8, ...$) that favour network's pruning (as suggested, in the past, by Hopfield himself and several other authors \cite{FAB,HopfieldUnlearning,unlearning0,unlearning1,VanHemmen} to erase spurious states).
\newline
The analysis of the information processing skills of this network has been accomplished elsewhere \cite{Albert,Complexity}: we summarize it by Fig.~$1$, referring to the original papers for further algorithmic details, while hereafter we deepen the mathematical aspects of its statistical mechanical foundation.

\begin{figure}
    \includegraphics[width=400pt]{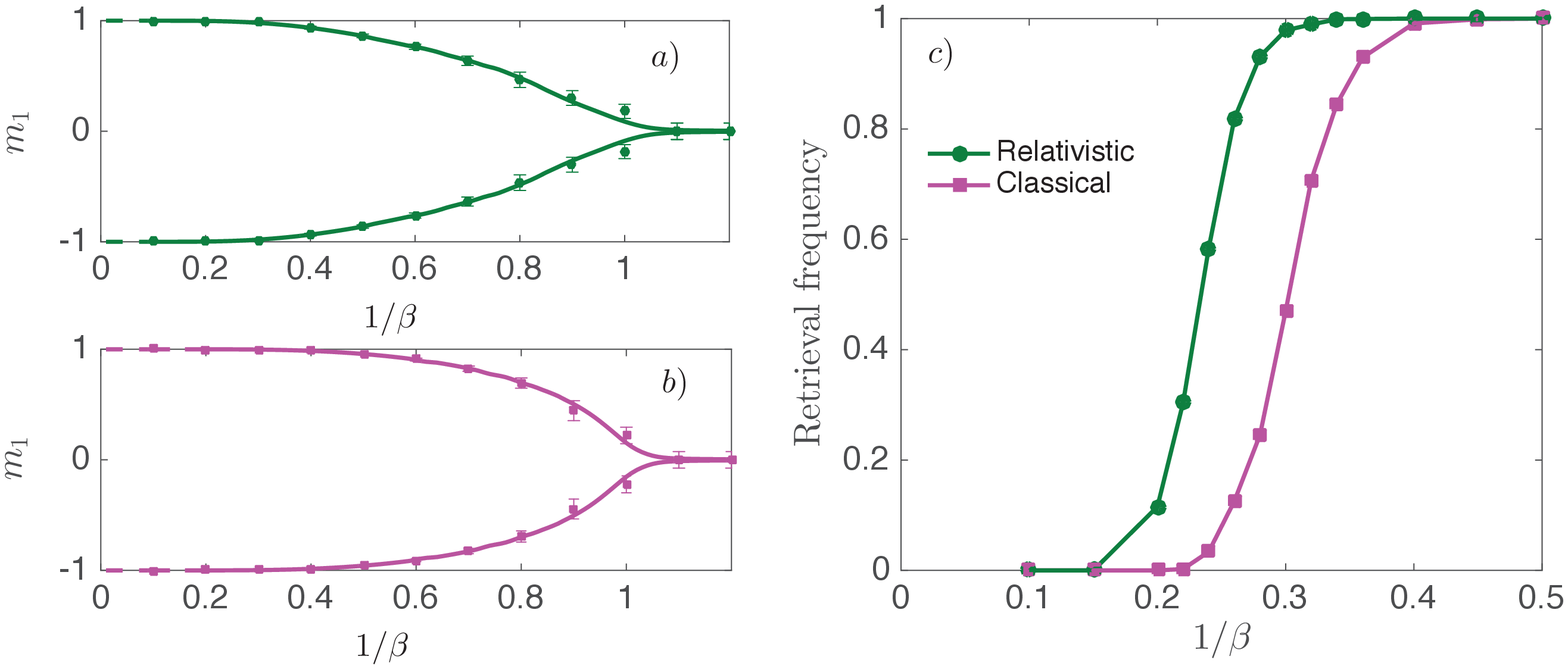}
      \caption{Left: Mattis overlap $m_1$ related to the retrieved patterns versus the noise level $1/\beta$ for the relativistic (panel $a$) and classic (panel $b$) Hopfield model. Symbols represent Monte Carlo runs for a network made of $N=400$ neurons and $P=3$ patterns while lines represent the numerical solution of the self-consistencies, respectively eq.~(\ref{self}) and eq.~(\ref{self-clax}). Right (panel $c$): Numerical comparison of the {\em retrieval frequency}, defined as the fraction of the runs that (randomly initiating the network dynamics) ended up into a pure state attractor, versus the noise level in the network for the relativistic ($\bigcirc$) and classic ($\square$) Hopfield model (again, $N=400$, $P=3$). }\label{FiguraUno} 
\end{figure}

\subsection{The statistical mechanical setting}
Having defined a suitable Hamiltonian playing as a cost function (see (\ref{Hopfield-classico}) and (\ref{Hopfield-relativistico})), the next step in the statistical-mechanics approach requires the introduction of a parameter $\beta \in \mathbb{R}^+$ to account for the noise level in the system (i.e., the {\em inverse temperature} in physical jargon) in such a way that for $\beta \to 0$ the network evolution is completely random, while in the $\beta \to \infty$ limit its evolution becomes deterministic (and the Hamiltonian plays as a Lyapounov function \cite{Coolen}). Then, we can state the next
\begin{Definition}
Defined $\mathbb{E}$ as the average over the patterns, we introduce the partition function $Z_N^{\textrm{(c,r)}}(\beta, \boldsymbol \xi)$ and the intensive free energy\footnote{Actually, we are committing a small abuse of notation as the ``standard'' free energy $\tilde{\alpha}(\beta)$ is defined as $\tilde{\alpha}(\beta)=-\beta^{-1}\alpha(\beta)$, but, as from the mathematical side $\tilde{\alpha}(\beta)$ and $\alpha(\beta)$ are (a constant apart) equivalent, we will keep our choice as it makes calculations more transparent, obviously, without affecting the results.} $\alpha_N^{\textrm{(c,r)}}(\beta)$ as
\begin{eqnarray}
Z_N^{\textrm{(c,r)}}(\beta, \boldsymbol \xi) &=& \sum_{\{ \sigma \}}^{2^N} e^{ - \beta H_N^{\textrm{(c,r)}}(\boldsymbol \sigma| \boldsymbol \xi)},\\
\label{eq:defalfa}
\alpha_N^{\textrm{(c,r)}}(\beta) &=& \frac{1}{N}\mathbb{E}\ln Z_N^{\textrm{(c,r)}}(\beta,  \boldsymbol \xi).
\end{eqnarray}
\end{Definition}
\begin{Definition}
Note that $\exp [ - \beta H_N^{\textrm{(c,r)}}(\boldsymbol \sigma| \boldsymbol \xi) ]/ Z^{\textrm{(c,r)}}(\beta,  \boldsymbol \xi)$ defines a probability measure, whose average (Boltzmann average) we indicate with $\langle \cdot \rangle_{\textrm{(c,r)}}$, i.e.
\be
\langle f(\boldsymbol \sigma)\rangle_{\textrm{(c,r)}} = \frac{\sum_{\{\sigma\}}^{2^N}f( \boldsymbol \sigma)e^{-\beta H^{\textrm{(c,r)}}_N(\boldsymbol \sigma |\boldsymbol \xi)}}{Z_N^{\textrm{(c,r)}}(\beta,  \boldsymbol \xi)}.
\ee
\end{Definition}
\begin{Remark} Note that to highlight  when observables are evaluated in the thermodynamic limit, we omit their subscript $N$. For instance, for the free-energy, we write
\be\label{alpha-limTD}
\alpha^{\textrm{(c,r)}}(\beta) = \lim_{N \to \infty}\alpha_N^{\textrm{(c,r)}}(\beta).
\ee
\end{Remark}
For the sake of completeness, we recall that for the standard Hopfield model in the infinite volume limit the free energy and the related self-consistency equations for the $P$ Mattis overlaps respectively read as \cite{Amit,Coolen,Viktor}
\begin{eqnarray}\label{free-clax}
\alpha^{\textrm{(c)}}(\beta)&=&\ln 2 + \left \langle \ln \cosh \left(\beta \boldsymbol{\xi}\cdot \langle \bold{m}\rangle_{\textrm{(c)}} \right) \right \rangle_{\boldsymbol \xi}  + \frac{\beta}{2} \langle \bold{m}^2 \rangle_{\textrm{(c)}},\\ \label{self-clax}
\langle m_{\mu}\rangle_{\textrm{(c)}} &=& \langle \xi^{\mu} \tanh\left( \beta  \boldsymbol{\xi}\cdot \langle \bold{m} \rangle_{\textrm{(c)}} \right)\rangle_{\boldsymbol \xi},
\end{eqnarray}
with the abbreviation $\langle f(\boldsymbol{\xi}) \rangle_{\boldsymbol \xi} = N^{-1} \sum_i f(\boldsymbol{\xi}_i)$.
\newline
In the following, we will focus solely on the relativistic generalization thus, in order to lighten the notation, we will drop the super- and sub-scripts $(r)$.

\section{Existence of the thermodynamic limit of the free energy} \label{sec:2}

In this section we prove the existence of $\alpha(\beta)$ for the relativistic Hopfield model. The underlying idea is to adapt the Guerra-Toninelli scheme, originally developed for a quadratic cost-function, to the Hamiltonian (\ref{Hopfield-relativistico}) featuring a square root. Following the standard scheme, we consider a system made of $N$ neurons and two other, independent, systems made of $N_1$ and $N_2$ neurons such that $N=N_1+N_2$. Then, we need to prove that the extensive free energy of the former is strictly smaller than the sum of those pertaining to the two subsystems (hence we have its sub-additivity), so that applying the Fekete Lemma \cite{fekete} we get the result \cite{Guerra-LimTerm1}. A way to compare the free energies of these two limiting cases is by interpolation, namely, defining an interpolating free energy whose extrema reproduce the free energies in the two cases of interest and then showing  that it derivative -w.r.t. the interpolating parameter- has semi-definite sign. Here the main adaptation will consist in a proof by reduction to absurd assuming the free energy to be super-additive (see Appendix A for the definitions of sub- super-additive successions and functions, and for the main statement of Fekete Lemma).

For the sake of simplicity (and without loss of generality) we fix $\beta =1$ just in this Section.
\newline
Let us consider, beyond the Mattis overlaps related to the $P$ patterns of the original $N$-neuron's model, also two further Mattis overlaps related to the two aforementioned smaller systems made of, respectively, $N_1$ and $N_2$ neurons. We can write
\begin{equation*}
\mathbf{m}_{N_{1}}:=\frac{1}{N_{1}}\sum_{i=1}^{N_{1}}\mathbf{\xi }_{i}\sigma
_{i},\quad \mathbf{m}_{N_{2}}:=\frac{1}{N_{2}}\sum_{i=1}^{N_{2}}\mathbf{\xi }%
_{i}\sigma _{i}
\end{equation*}%
and notice that
\begin{eqnarray*}
\mathbf{m}_{N} &=&\frac{1}{N}\sum_{i=1}^{N}\mathbf{\xi }_{i}\sigma _{i}=%
\frac{1}{N}\sum_{i=1}^{N_{1}+N_{2}}\mathbf{\xi }_{i}\sigma _{i}=\frac{1}{N}%
\left( \sum_{i=1}^{N_{1}}\mathbf{\xi }_{i}\sigma _{i}+\sum_{j=1}^{N_{2}}%
\mathbf{\xi }_{j}\sigma _{j}\right) \\
&=&\frac{1}{N}\left( N_{1}\mathbf{m}_{N_{1}}+N_{2}\mathbf{m}_{N_{2}}\right)
=\rho _{1}\mathbf{m}_{N_{1}}+\rho _{2}\mathbf{m}_{N_{2}},
\end{eqnarray*}%
once introduced the relative densities $\rho_1, \rho_2$ as
\begin{equation*}
\rho _{1}:=\frac{N_{1}}{N}\quad \text{and\quad }\rho _{2}:=\frac{N_{2}}{N}.
\end{equation*}

\begin{Definition}
Let us introduce an interpolating parameter $t \in \left[0,1\right]$ that we use to define an interpolating free energy $\alpha_N(\beta,t)$ as follows
\begin{equation}\label{alpha-inter}
\alpha_N \left( \beta ,t\right) =\frac{1}{N} \mathbb{E}\ln \sum_{\left\{ \sigma \right\}
}^{2^{N}}\exp \left[ tN\sqrt{1+\mathbf{m}_{N}^{2}}+\left( 1-t\right) \left(
N_{1}\sqrt{1+\mathbf{m}_{N_{1}}^{2}}+N_{2}\sqrt{1+\mathbf{m}_{N_{2}}^{2}}%
\right) \right].
\end{equation}%
\end{Definition}
It is crucial to observe that in the two limits of $t \to 1$ and $t \to 0$ we recover, respectively,
\begin{equation}
\alpha_N \left( \beta ,1\right) =\frac{1}{N}  \mathbb{E}\ln \sum_{\left\{ \sigma \right\}
}^{2^{N}}\exp \left( N\sqrt{1+\mathbf{m}_{N}^{2}}\right) =\alpha _{N}\left(
\beta \right) ,  \label{34}
\end{equation}
while
\begin{eqnarray}
\alpha_N \left( \beta ,0\right) &=&\frac{1}{N}  \mathbb{E}\ln \sum_{\left\{ \sigma
\right\} }^{2^{N}}\exp \left( N_{1}\sqrt{1+\mathbf{m}_{N_{1}}^{2}}+N_{2}%
\sqrt{1+\mathbf{m}_{N_{2}}^{2}}\right)  \notag \\
&=&\frac{1}{N}  \mathbb{E}\ln \left( Z_{N_{1}}Z_{N_{2}}\right) =\frac{1}{N}\mathbb{E}\left( \ln
Z_{N_{1}}+\ln Z_{N_{2}}\right)  \notag \\
&=&\rho _{1}\alpha _{N_{1}}\left( \beta \right) +\rho _{2}\alpha
_{N_{2}}\left( \beta \right) .  \label{35}
\end{eqnarray}%
Therefore, by varying $t$, we interpolate between the original system and the sum of the two smaller subsystems, properly weighted by their relative densities. By the Fundamental Theorem of Calculus we can write
\begin{equation}
\alpha_N \left( \beta ,1\right) =\alpha_N \left( \beta ,0\right) +\int_{0}^{1}%
\frac{d \alpha_N \left( \beta ,t\right) }{d t}dt,  \label{33}
\end{equation}%
and, since the {\em integral} operator is monotonous (and therefore it must respect inequalities), it will be sufficient to prove that the derivative of the interpolating free energy w.r.t. $t$, i.e., $\dot{\alpha}_N \left( \beta ,t\right) \equiv d_t \alpha_N(\beta,t)$, has a negative semi-definite sign.
\newline
Denoting with $\tilde{Z}$ the interpolating partition function coupled to the interpolating free energy (\ref{alpha-inter}) we can write
\begin{eqnarray}
\frac{d \alpha_N \left( \beta ,t\right) }{d t} &=&\frac{\partial
}{\partial t}\frac{1}{N} \mathbb{E}\ln \sum_{\left \{ \sigma \right\} }^{2^{N}}\exp
\left [ tN\sqrt{1+\mathbf{m}_{N}^{2}}+\left( 1-t\right) \left ( N_{1}\sqrt{1+%
\mathbf{m}_{N_{1}}^{2}}+N_{2}\sqrt{1+\mathbf{m}_{N_{2}}^{2}}\right) \right]
\notag \\
&=&\frac{1}{N}\sum_{\{ \sigma \}}^{2^N} \mathbb{E} \frac{1}{\tilde{Z}}\left( N\sqrt{1+\mathbf{m}_{N}^{2}}-N_{1}%
\sqrt{1+\mathbf{m}_{N_{1}}^{2}}+N_{2}\sqrt{1+\mathbf{m}_{N_{2}}^{2}}\right)
\exp \left ( \phi \right ) ,  \label{30}
\end{eqnarray}%
where we called
\begin{equation*}
\phi :=tN\sqrt{1+\mathbf{m}_{N}^{2}}+\left( 1-t\right) \left ( N_{1}\sqrt{1+%
\mathbf{m}_{N_{1}}^{2}}+N_{2}\sqrt{1+\mathbf{m}_{N_{2}}^{2}}\right ) .
\end{equation*}%
\begin{Remark}
We observe that eq. (\ref{alpha-inter}) allows us to generalize the partition function $Z(\beta, \boldsymbol \xi)$ to $Z(\beta,\boldsymbol \xi, t)$, such that eq. (\ref{30}) can be written in terms of its $t$-generalized Boltzmann average $\langle . \rangle_t$  as
\be
\frac{d \alpha_N \left( \beta ,t\right) }{d t}  = \left\langle \sqrt{1+\mathbf{m}_{N}^{2}}-\rho _{1}\sqrt{1+\mathbf{m}%
_{N_{1}}^{2}}-\rho _{2}\sqrt{1+\mathbf{m}_{N_{2}}^{2}}\right\rangle _{t}.
\label{32}
\ee
\end{Remark}
Now we must prove that the expression averaged in the r.h.s. of eq. (\ref{32}) is less or equal to zero, as stated by the next
\begin{Proposition}
\label{proposizioneModelloRelativisticoSubadditivo} The $t$-derivative of the interpolating free energy (\ref{alpha-inter}) is semi-definite negative, namely
\begin{equation}
\label{aste}
\sqrt{1+\mathbf{m}_{N}^{2}}-\rho _{1}\sqrt{1+\mathbf{m}_{N_{1}}^{2}}-\rho
_{2}\sqrt{1+\mathbf{m}_{N_{2}}^{2}}\leq 0.
\end{equation}
\end{Proposition}

\begin{proof}
We first observe that, as $\rho _{2}=1-\rho _{1}=1-\rho$, the parameters effectively involved in eq. (\ref{32}) are solely $\mathbf{m}_{N_{1}}$, $\mathbf{m}_{N_{2}}$ and $\rho _{1}$, that from now on we will call simply $\rho$, in such a way that $ \mathbf{m}_{N}=\rho _{1}\mathbf{m}_{N_{1}}+\rho _{2}\mathbf{m}_{N_{2}}=\rho
\mathbf{m}_{N_{1}}+\left( 1-\rho \right) \mathbf{m}_{N_{2}}$.
Then, we can write the l.h.s. of eq. (\ref{aste}) as $
\sqrt{1+\left( \rho \mathbf{m}_{N_{1}}+\left( 1-\rho \right) \mathbf{m}%
_{N_{2}}\right) ^{2}}-\rho \sqrt{1+\mathbf{m}_{N_{1}}^{2}}-\left( 1-\rho
\right) \sqrt{1+\mathbf{m}_{N_{2}}^{2}}$.
Let us now suppose, by contradiction, that
\begin{equation*}
\sqrt{1+\left( \rho \mathbf{m}_{N_{1}}+\left( 1-\rho \right) \mathbf{m}%
_{N_{2}}\right) ^{2}}-\rho \sqrt{1+\mathbf{m}_{N_{1}}^{2}}-\left( 1-\rho
\right) \sqrt{1+\mathbf{m}_{N_{2}}^{2}}>0.
\end{equation*}%
With some algebra we obtain
\begin{equation*}
2\rho \left( 1-\rho \right) >2\rho \left( 1-\rho \right) \left( \sqrt{\left(
1+\mathbf{m}_{N_{1}}^{2}\right) \left( 1+\mathbf{m}_{N_{2}}^{2}\right) }-%
\mathbf{m}_{N_{1}}\mathbf{m}_{N_{2}}\right),
\end{equation*}%
and, as $2\rho \left( 1-\rho \right) \neq 0$, we have
$
1>\sqrt{\left( 1+\mathbf{m}_{N_{1}}^{2}\right) \left( 1+\mathbf{m}%
_{N_{2}}^{2}\right) }-\mathbf{m}_{N_{1}}\mathbf{m}_{N_{2}},
$
whence
\begin{equation*}
\sqrt{\left( 1+\mathbf{m}_{N_{1}}^{2}\right) \left( 1+\mathbf{m}%
_{N_{2}}^{2}\right) }<1+\mathbf{m}_{N_{1}}\mathbf{m}_{N_{2}}.
\end{equation*}%
The r.h.s. term of the above expression is certainly non negative
(as $\mathbf{m}_{N_{1}},\mathbf{m}_{N_{2}}\in \left[
-1,1\right]$): even in the worst scenario where $\mathbf{m}_{N_{1}}$ and $\mathbf{m}%
_{N_{2}}$ have opposite signs, their product will never be smaller than $-1$ hence, by quadrature the inequality remains unchanged.
We are then allowed to write
\begin{equation*}
\left( 1+\mathbf{m}_{N_{1}}^{2}\right) \left( 1+\mathbf{m}%
_{N_{2}}^{2}\right) <\left( 1+\mathbf{m}_{N_{1}}\mathbf{m}_{N_{2}}\right)
^{2}
\end{equation*}%
by which we get
\begin{equation*}
\left( \mathbf{m}_{N_{2}}^{2}-\mathbf{m}_{N_{1}}^{2}\right) ^{2}<0
\end{equation*}%
that is obviously wrong.\hfill
\end{proof}
\begin{Remark}
We also highlight that the function $f:x\mapsto \sqrt{1+x^{2}}$ is convex, namely, for $\lambda \in \left[ 0,1\right]$
\begin{equation*}
f\left( \lambda x_{1}+\left( 1-\lambda \right) x_{2}\right) \leq \lambda
f\left( x_{1}\right) +\left( 1-\lambda \right) f\left( x_{2}\right) .
\end{equation*}%
If we identify%
\begin{equation*}
\lambda :=\rho, \ \quad x_{1}=\mathbf{m}_{N_{1}}, \ \quad x_{2}=\mathbf{m}_{N_{2}},
\end{equation*}%
as $\mathbf{m}_{N}=\rho \mathbf{m}_{N_{1}}+ (1-\rho) \mathbf{m}_{N_{2}}$, then
\begin{equation*}
\sqrt{1+\mathbf{m}_{N}^{2}}=\sqrt{1+\left( \rho _{1}\mathbf{m}_{N_{1}}+\rho
_{2}\mathbf{m}_{N_{2}}\right) ^{2}}\leq \rho _{1}\sqrt{1+\mathbf{m}%
_{N_{1}}^{2}}+\rho _{2}\sqrt{1+\mathbf{m}_{N_{2}}^{2}}.
\end{equation*}
\end{Remark}

Proposition (\ref{proposizioneModelloRelativisticoSubadditivo})  allows us to state that
\begin{equation*}
\alpha_N \left( \beta ,1\right) -\alpha_N \left( \beta ,0\right) =\int_{0}^{1}%
\frac{\partial \alpha_N \left( \beta ,t\right) }{\partial t}dt\leq 0
\end{equation*}
namely
\begin{equation}
N \alpha_{N}\left( \beta \right) \leq N _{1}\alpha _{N_{1}}\left( \beta
\right) + N_{2}\alpha _{N_{2}}\left( \beta \right),  \label{36}
\end{equation}%
such that we can finally state the next
\begin{Theorem}
The infinite volume limit of the intensive free energy defined by the relativistic Hopfield cost function (\ref{Hopfield-relativistico}) exists and it equals its infimum, that is
\begin{equation*}
\exists \lim_{N\rightarrow \infty }\alpha _{N}\left( \beta \right)
=\inf_{N\in \mathbb{N}}\left\{ \alpha _{N}\left( \beta \right) \right\}
=\alpha \left( \beta \right) .
\end{equation*}
\end{Theorem}
\begin{proof}
The proof is a straightforward consequence of Proposition $1$ coupled to the Fekete Lemma. \hfill
\end{proof}

\section{Expression of the thermodynamic limit of the free energy} \label{sec:3}

In this Section we give an explicit expression for the infinite volume limit of the free energy related to the relativistic Hopfield cost-function (\ref{Hopfield-relativistico}) in terms of the Mattis overlaps. Again, we use a suitable Guerra's interpolation scheme, in such a way that one extremum of the interpolation recovers the original model to be solved and the other recovers a case whose solution is straightforward. In particular, as we know that all the mean-field models (even the generalized ones) have a product space structure, namely their probability distribution $P(\sigma_1,\sigma_2,...,\sigma_N)$ in the thermodynamic limit factorizes (i.e., $\lim_{N \to \infty} P(\sigma_1,\sigma_2,...,\sigma_N) = \prod_{i=1}^{N}P(\sigma_i)$), we can use this information to construct the ``easy'' extremum of the interpolation. In other words,
\begin{Definition}
Being the scalar $t \in [0,1]$ an interpolating parameter, we define a novel interpolating free energy $\alpha_N(\beta,t)$ as
\be\label{alpha-intes}
\alpha_N(\beta,t) := \frac{1}{N} \mathbb{E} \ln   \sum_{ \{ \sigma \}}^{2^N}
\exp\left( t \beta N \sqrt{1 + \sum_{\mu=1}^{P}m_{\mu}^2} + (1-t) \beta \sum_{\mu=1}^{P} \psi^{\mu} \sum_{i=1}^{N} \xi_i^{\mu}\sigma_i \right),
\ee
where $\psi^{\mu}$, $\mu \in (1,...,P)$, are $P$ fields that are functions of the neurons and of the patterns.
\end{Definition}
It is worth stressing that each neuron experiences the simultaneous action of $P$ fields $\{ \psi^{\mu}\}_{\mu=1,...,P}$ linearly combined, unlike the interpolating structure working in the ferromagnetic case where each spin just experiences the action of one single field \cite{Barra0}.
The particular choice of the fields $\psi^{\mu}$ will be discussed later.
\begin{Remark}
Given a smooth function $F(\boldsymbol \sigma|\boldsymbol \xi)$ of the neurons and of the patterns, we observe that the interpolating free-energy (\ref{alpha-intes}) implicitly defines the extended averages $\langle F(\boldsymbol \sigma| \boldsymbol\xi) \rangle_t$ as
\be\label{extended-average}
\langle F(\boldsymbol \sigma| \boldsymbol \xi) \rangle_t := \frac{\sum_{\{ \sigma \}}^{2^N} F(\boldsymbol \sigma| \boldsymbol \xi) \exp\left( t \beta N \sqrt{1 + \sum_{\mu=1}^{P}m_{\mu}^2} + (1-t) \beta \sum_{\mu=1}^{P} \psi^{\mu} \sum_{i=1}^{N} \xi_i^{\mu}\sigma_i \right) }{\sum_{\{ \sigma\}}^{2^N}\exp\left( t \beta N \sqrt{1 + \sum_{\mu=1}^{P}m_{\mu}^2} + (1-t) \beta \sum_{\mu}^{P} \psi^{\mu} \sum_{i=1}^{N} \xi_i^{\mu}\sigma_i \right)}.
\ee
\end{Remark}
We now must identify the extrema of $\alpha_N(\beta,t)$: at $t=1$ we get the original model we aim to solve, while at $t=0$ we get a system with one-body terms. The latter can be handled easily and its contribution to the free energy reads out as
\be\label{Cauchy}
\alpha_N(\beta,t=0) = \ln 2 + \left \langle \ln\cosh\left( \beta \sum_{\mu=1}^{P}\psi^{\mu} \langle m_{\mu}\rangle \right) \right \rangle_{\boldsymbol{\xi}}.
\ee
To obtain the free energy of the model $\alpha_N(\beta)=\alpha_N(\beta,t=1)$,  we use again the Fundamental Theorem of Calculus and write
\be
\alpha(\beta)=\lim_{N \to \infty} \alpha_N(\beta,t=1) =\lim_{N \to \infty}\left( \alpha_N(\beta,t=0) + \int_{0}^{1}\dot{\alpha}_N(\beta,t')dt' \right),
\ee
where, as usual, $\dot{\alpha}_N(\beta,t) \equiv d_t \alpha_N(\beta,t)$ is the $t$-derivative of the interpolating free energy. The latter reads as
\be\label{punto}
\frac{d \alpha_N(\beta,t)}{dt} = \beta \left \langle \left( \sqrt{1+ \langle\bold{m}_N^2\rangle}- \sum_{\mu=1}^{P}\psi^{\mu}\langle m_{\mu}\rangle \right) \right \rangle_{\boldsymbol{\xi}}.
\ee
In general, its evaluation is an hard task, yet in the thermodynamic limit (where we are focusing), calling $M_{\mu}$ the limiting value of the $\mu$-th Mattis overlap (i.e., $\lim_{N \to \infty}P(\bold{m})=\delta(\bold{m}-\bold{M})$), by requiring $\beta$-almost-surely that
\begin{eqnarray}\label{selfaveragingmM}
&& \lim_{N \to \infty} \sum_{\mu=1}^{P} \left \langle \left( m_{\mu}-M_{\mu} \right)^2 \right \rangle =0,\\
&& \lim_{N \to \infty} \left \langle \left(\sqrt{1+ \bold{m^2}}-\sqrt{1+\bold{M}^2} \right)^2 \right \rangle =0,
\end{eqnarray}
namely the self-averaging of the order parameters  and of the energy of the model \cite{Barra0,Barra-Bipartiti,Guerra-SumRules}, as $\langle m_{\mu}\rangle \to M_{\mu}$ we obtain
$$
\left \langle\sqrt{1+\bold{m^2}} - \frac{\bold{m} \cdot \bold{M}}{\sqrt{1+\bold{M}^2}} + \frac{1}{\sqrt{1+\bold{M}^2}}  \right \rangle=0.
$$
Comparing the above equation with the r.h.s. of (\ref{punto}) and choosing $\psi^{\mu}=\frac{M_{\mu}}{\sqrt{1+\bold{M}^2}}$, we can write
$$
\frac{d \alpha(\beta,t)}{dt} - \frac{\beta}{\sqrt{1+\bold{M}^2}}=0.
$$
Merging the last result with the Cauchy condition (\ref{Cauchy}) we can finally state the next
\begin{Theorem}
The infinite volume limit of the free energy for the relativistic Hopfield cost function (\ref{Hopfield-relativistico}) reads, in terms of its $P$ Mattis overlaps, as
\be\label{free-final}
\alpha(\beta)=\ln 2 + \left \langle \ln \cosh \left(\beta \boldsymbol{\xi}\cdot \frac{ \bold{M}}{\sqrt{1+  \bold{M}^2 }} \right) \right \rangle_{\boldsymbol \xi} + \frac{\beta}{\sqrt{1+  \bold{M}^2}},
\ee
and its related self-consistency equation reads as
\be\label{self}
M_{\mu}  = \left \langle \xi^{\mu} \tanh\left( \beta \frac{\sum_{\mu=1}^P \xi^{\mu} M_{\mu}}{\sqrt{1+\sum_{\mu=1}^{P}M_{\mu}^2}} \right) \right \rangle_{\boldsymbol \xi}
\ee
in agreement with the previous results reported in \cite{Albert}.
\end{Theorem}
\begin{Remark}
We stress that the leading order expansion of the above free energy (\ref{free-final}) and self-consistency (\ref{self}) correctly recovers the standard Hopfield picture coded by eq.s (\ref{free-clax}) and (\ref{self-clax}).
\end{Remark}
We can now move on to study the critical behavior of the system. To this task it will be convenient to introduce the following
\begin{Definition}
Taking, without loss of generality, $\boldsymbol{\xi}^1$ as the candidate pattern to be retrieved, we introduce its centered and rescaled Mattis overlap as
\be\label{riscalati}
\hat{m}_1 := \sqrt{N}\left(m_1 - M_1 \right),
\ee
where, as usual,  $m_{1} := N^{-1}\sum_{i=1}^{N}\xi_i^{1}\sigma_i$ while $M_{1}$ is its thermodynamic limit, namely $\lim_{N\to\infty}\langle m_{1}\rangle \to M_{1}$.
\end{Definition}
Notice that $\langle \hat{m}_1^2 \rangle$ scales as the variance of $\langle m_1 \rangle$ times the system size $N$. In this way, we can investigate the critical behaviour of the system by studying $\langle \hat{m}_1^2 \rangle$ as a function of the noise $\beta$: we look for those values $\beta_c$ where the fluctuations of the order parameters diverge as this is a signature of criticality.
\newline
In order to get an explicit expression for $\langle \hat{m_1^2}\rangle$, we exploit again the interpolation scheme (\ref{alpha-intes}) and we write $\langle \hat{m}_1^2 \rangle \equiv \langle \hat{m}_1^2\rangle_{t=1}$ as
\be
\langle \hat{m}_1^2\rangle_{t=1} = \langle \hat{m}_1^2\rangle_{t=0} + \int_{0}^{1} \frac{d \langle \hat{m}_1^2\rangle_{t'}}{dt'}dt,
\ee
where the generalized average $\langle \cdot \rangle_t$ was defined in (\ref{extended-average}).\\
Note that our calculation start at $t=0$ where the system is decoupled: it is thus both natural and convenient to approach the critical line from the high noise region as in this region it is possible to use standard central-limit-theorem-like arguments to assume the probability distribution of the $\hat{m}_1$ to be a Gaussian\footnote{This assumption crucially allows us to use the Wick Theorem, namely, given a Gaussian variable $z$ we can write $\mathbb{E}(z \cdot f(z))=\mathbb{E}(z^2)\mathbb{E}(\partial_z f(z))$.}: the Cauchy condition $\langle \hat{m}_1^2 \rangle_{t=0}$  (that is the standard {\em high-temperature} value) can be easily shown to be $\langle \hat{m}_1^2 \rangle_{t=0}=1$ as $\langle \hat{m}_1^2 \rangle_{t=0} = \langle N \left( m_{1} - M_{1} \right)^2 \rangle_{t=0} = 1 + (N-1) M_1^2 + N M_1^2 - 2N M_1^2=1-M_1^2=1$ as in the ergodic region trivially $M_1^2=0$).
We must now face the $t$-derivative: to this task it is useful to state the next
\begin{Proposition}
Retain $\psi_{\mu}=M_{\mu}/(\sqrt{1+\sum_{\mu=1}^{P} M_{\mu}^2})$ and let $F$ be a smooth function of the Mattis overlaps, then, above and close to the critical point (namely where the Mattis overlaps are zero or infinitesimal) the following streaming equation holds: 
\be\label{streaming} 
\frac{d}{dt} \langle F \rangle_t \sim \frac{\beta}{2}\left( \langle F \bold{\hat{m}}^2 \rangle_t - \langle F \rangle \langle \bold{\hat{m}}^2 \rangle_t \right).
\ee
\end{Proposition}
\begin{proof}
The proof uses the fact that we are inspecting the critical behavior (hence we assume the Mattis overlap to move from zero continuously when reached a critical noise level) and works by direct brute force:
\be\label{bruttaforza}
\frac{d}{dt} \langle F \rangle_t = \beta N \left( \langle F \sqrt{1+\bold{m}^2} \rangle_t - \sum_{\mu=1}^{P}\frac{M_{\mu}}{\sqrt{1+\bold{M}^2}}\langle F m_{\mu}\rangle_t - \langle F \rangle_t \langle \sqrt{1+\bold{m}^2} \rangle_t + \frac{M_{\mu}}{\sqrt{1+\bold{M}^2}} \langle F \rangle_t \langle m_{\mu}\rangle_t  \right).
\ee
Expanding $\sqrt{1+x^2} \sim 1 + x^2/2$ and $1/(\sqrt{1+x^2}) \sim 1 - x^2/2$ in the above equation and adding and subtracting twice $(\beta N/2)\langle F \rangle_t \hat{m}^2$ we have the result. \hfill
\end{proof}
The above Proposition is the key to state the last
\begin{Theorem}
The centered and rescaled fluctuations of the Mattis overlap $\langle m_1 \rangle$ (quantifying the retrieval of pattern $\bold{\xi}^1$) behaves as
\be\label{fluct-theor}
\langle \hat{m}_1^2 \rangle = \langle \hat{m}_1^2 \rangle_{t=1} = \frac{1}{1-\beta},
\ee
thus the high-noise (i.e. {\em ergodic} to keep the physical jargon) region is limited to $\beta < \beta_c \equiv 1$.
\newline
In the low-noise (i.e. {\em broken-symmetry}) region the Mattis overlap may assume non-null values (see Fig.~$1$).
\end{Theorem}
\begin{proof}
Applying the streaming equation (\ref{streaming}) on $\langle \hat{m}_1^2 \rangle$ we get the following Cauchy-problem
\begin{eqnarray}
\frac{d}{dt}\langle \hat{m}_1^2 \rangle_t &=& \beta \langle \hat{m}_1^2 \rangle^2,\\
\langle \hat{m}_1^2 \rangle_{t=0} &=& 1,
\end{eqnarray}
whose solution is exactly eq. (\ref{fluct-theor}).
\hfill
\end{proof}

\section{Conclusions and outlooks} \label{sec:4}

In the last decade Artificial Intelligence has permeated our everyday lives and, as a natural consequence, the quest for novel and stronger tools to handle with it has become a urgent priority in different fields of Science \cite{DL1}. In particular, within the statistical mechanics route, heuristic and rigorous results appeared steadily in the Theoretical and Mathematical Physics Communities in the last years, typically with heuristic approaches first (as in the case of the standard Hopfield model whose analysis was pioneered by Amit, Gutfreund and Sompolinsky by replica-trick techniques \cite{Amit} and later confirmed by more rigorous techniques, see e.g., \cite{Agliari-Barattolo,Bovier1,Tala1}).

Along these lines, in this paper, relying on Guerra's interpolation schemes \cite{Guerra-LimTerm1,Guerra-SumRules}, we have rigorously analyzed the statistical mechanical picture of a new associative neural network, that is, the {\em relativistic Hopfield model} recently introduced in \cite{Albert}. More precisely, by an adaptation of the classical Guerra-Toninelli argument, we have proved that the infinite volume limit of the free energy exists and it is well defined for this model; we also gave its explicit expression in terms of the $P$ Mattis overlaps, namely the natural order parameters of the theory, that perfectly matches the results of \cite{Albert}. Once extremized the free energy over the Mattis overlaps, the $P$ self-consistent equations for the order parameters have also been obtained and, with them, a picture of the critical phase transition that the model undergoes when the noise level $\beta$ crosses the critical value $\beta_c=1$. To obtain the explicit expression for the free energy we required the self-averaging (in the infinite volume limit) of the overlaps: this assumption has been confirmed immediately after, by  inspecting their (centered and rescaled) fluctuations. The divergence of these fluctuations is found to happen at $\beta_c=1$: this result highlights a critical behavior similar to the one known for the standard Hopfield model in the low-storage regime, despite the presence of many-body contributions.
\newline
From a physical viewpoint, this may, at a first glance, sound weird -as $P$-spin systems are known to exhibit discontinuous phase transitions (see e.g., \cite{BarraPspin})- however an intuitive explanation for this is that, keeping the r.h.s. of eq. (\ref{series}) in mind, the pairwise-interaction provides the main driving force toward a broken phase, but the next one, i.e., the fourth-order in $\sigma$ (which is anyway small if compared to the previous order as the series (\ref{series}) obviously converges) appears with the opposite sign thus it drives the system toward a non-magnetized state $\bold{M} \to 0$. Hence we have to wait the even smaller contribution, i.e., sixth-order in $\sigma$, to have another strengthening toward a magnetized state, but this is actually negligible.
\newline
From a mathematical viewpoint, instead, the request (\ref{selfaveragingmM}) of the self-averaging of the order parameters gives also a hint that, for the high storage analysis, different approaches would be necessary. This can be understood with a scaling argument by observing that the request (\ref{selfaveragingmM}) breaks down if $P \propto N$ as the term inside the parenthesis scales as $N^{-1}$.
\newline
Finally, it is worth pointing out that the outlined mathematical scaffold stays robust against pattern's dilution, namely for multitasking networks \cite{Agliari-PRL1}, whose {\em relativistic extension} we plan to report soon.

\section*{Acknowledgments}
\noindent
A.B. and M.N. acknowledge Salento University. \\
E.A. and A.B. are grateful to GNFM-INdAM ({\em Progetto Giovani 2018}) for financial support.\\
E.A. acknowledges the grant {\em Progetto Ateneo} (RG11715C7CC31E3D) from Sapienza University of Rome.\\
A.B. also acknowledges MIUR (through basic funding to the Italian research), INFN and the grant {\em Rete Match: Progetto Pythagoras} (CUP:J48C17000250006) for financial support.

\section*{Appendix A: sub/super-additivity and Fekete Lemma}

\begin{Definition}[sub-additive, super-additive functions]
A function $f:A\rightarrow B$ is termed {\em sub-additive} (resp. {\em super-additive}) if
\begin{equation*}
\forall x,y\in A,\qquad f\left( x+y\right) \leq f\left( x\right) +f\left(
y\right)
\end{equation*}%
(resp.%
\begin{equation*}
\forall x,y\in A,\qquad f\left( x+y\right) \geq f\left( x\right) +f\left(
y\right) ~\text{).}
\end{equation*}
\end{Definition}
In the same way we can define the concept of sub-additive succession
in the following way:
\begin{Definition}[sub-additive, super-additive successions]
A succession $\left( a_{n}\right) _{n\in \mathbb{N}}$ is termed {\em sub-additive}
(resp. {\em super-additive}) if, $\forall n,m \in \mathbb{N}$
\begin{equation*}
a_{n+m}\leq a_{n}+a_{m}
\end{equation*}%
(risp.%
\begin{equation*}
a_{n+m}\geq a_{n}+a_{m}~\text{).}
\end{equation*}
\end{Definition}
\begin{Lemma}
\label{lemmaDiseguaglianza} Let $\left( a_{n}\right) _{n\in \mathbb{N}}$ be a sub-additive succession, then the follwing holds
\begin{equation*}
a_{nk}\leq ka_{n}\quad \forall k\in \mathbb{N}
\end{equation*}
\end{Lemma}
\begin{proof}
A sketched proof is immediate by using the Induction Principle over $k$. Indeed, let $\left( a_{n}\right) _{n\in \mathbb{N}}$ be a sub-additive succession. If $k=1$ then
\begin{equation*}
a_{n}\leq 1a_{n}=a_{n}.
\end{equation*}%
Now let us pose $k>1$ and assume the thesis still hold for any natural up to $k-1$ and let us prove it for $k$ as well. To be true at $k-1$ it must be that
\begin{equation}
a_{n\left( k-1\right) }\leq \left( k-1\right) a_{n}  \label{2.0}
\end{equation}%
and thus that
\begin{equation}
a_{n\left( k-1\right) }+a_{n}\leq ka_{n}.  \label{2}
\end{equation}%
Note that, thanks to the sub-additivity property of the succession, we can write
\begin{equation}
a_{nk-n}+a_{n}\geq a_{nk-n+n}=a_{nk}.  \label{3}
\end{equation}%
Merging inequalities (\ref{2}) and (\ref{3}) the thesis follows trivially. \hfill
\end{proof}

\begin{Proposition}[Lemma di Fekete]
\label{lemmaFekete}For any sub-additive succession $\left( a_{n}\right)
_{n\in \mathbb{N}}$%
\begin{equation*}
\exists \lim_{n\rightarrow \infty }\frac{a_{n}}{n}=\inf_{n\in \mathbb{N}}%
\frac{a_{n}}{n}.
\end{equation*}
Similar conclusions can be drawn for super-additive successions as well and in this case
case we have that
\begin{equation*}
\exists \lim_{N\rightarrow +\infty }\frac{a_{n}}{n}=\sup_{n\in \mathbb{N}}%
\frac{a_{n}}{n}
\end{equation*}
\end{Proposition}

\end{document}